\documentclass{article}
\usepackage{fullpage}
\usepackage{graphicx}
\usepackage{algorithm}
\usepackage{algpseudocode}
\usepackage{amssymb,amsmath,amsthm}
\usepackage{cleveref}
\usepackage{algpseudocode}
\usepackage{pdfpages}
\usepackage{placeins}
\usepackage{subcaption}
\newtheorem{theorem}{Theorem}
\newtheorem{lemma}{Lemma}

\newcommand{\norm}[2][p]{\left\lVert#2\right\rVert_{#1}}
\newcommand{\vv}[1]{\overrightarrow{#1}}
\newcommand{\Epsilon}{E}
\DeclareMathOperator{\Length}{Len}

\DeclareMathOperator{\poly}{poly}
\DeclareMathOperator{\G}{G}
\DeclareMathOperator{\T}{T}
\newcommand{\Len}[1]{\Length(#1)}

\begin{document}

\title{Approximating The $p$-Mean Curve of Large Data-Sets}

\author{Sepideh Aghamolaei\footnote{Department of Computer Engineering,
        Sharif University of Technology, Tehran, Iran. aghamolaei@ce.sharif.edu} \and
Mohammad Ghodsi\footnote{Department of Computer Engineering,
        Sharif University of Technology, Tehran, Iran. 
        School of Computer Science, Institute for Research in Fundamental Sciences (IPM), Tehran, Iran. ghodsi@sharif.edu}
}
\date{}

\maketitle

\begin{abstract}
A set of piecewise linear functions, called polylines, $P_1,\ldots,P_L$ each with at most $n$ vertices can be simplified into a polyline $M$ with $k$ vertices, such that the Fr\'echet distances $\epsilon_1,\ldots,\epsilon_L$ to each of these polylines are minimized under the $L_p$ distance. We call $M$ for $L_p$ with $p\geq 1$ a $p$-mean curve ($p$-MC).

We discuss $p\geq 1$, for which $L_p$ distance satisfies the triangle inequality and $p$-mean has not been discussed before for most values $p$.
Computing the $p$-mean polyline is NP-hard for $L=\Omega(1)$ and some values of $p$, so we discuss approximation algorithms.

We give a $O(n^2\log k)$ time exact algorithm for $L=2$ and $p\geq 1$. Also, we reduce the Fr\'echet distance to the discrete Fr\'echet distance which adds a factor $2$ to both $k$ and $\epsilon$.
Then we use our exact algorithm to find a $3$-approximation for $L>2$ in $\poly(n,L)$ time. Our method is based on a generalization of the free-space diagram (FSD) for Fr\'echet distance and composable core-sets for approximate summaries.
\end{abstract}

\section{Introduction}
A {\em polygonal curve} is a sequence of points, e.g. GPS data such as vehicle tracks on a map, time series, movement patterns, or discretized borders of countries on a map. Trajectories appear in spatial databases and networks, geographic information systems (GIS), and any dataset with temporal labels for coordinates. Simplification is a method of reducing the size of the input trajectory, mostly to achieve reduced noise, optimize the storage space, or as a preprocessing step to improve the running time of later processing algorithms.

For large datasets and models for them, such as streaming, divide and conquer including massively parallel computations (MPC) \cite{beame2013communication}, MapReduce class (MRC) \cite{mrc}, and {\em composable core-sets} \cite{composable}, there are few algorithms with good theoretical guarantees. Even on medium-sized datasets, existing algorithms take at least quadratic time for some similarity measures, and are therefore too slow to be useful in practice. Methods for partitioning data while keeping the theoretical guarantees and relaxing the condition of the simplified curve be built from the points of the original curve are our main tools in achieving these goals.

The {\em min-k curve simplification} problem finds a subcurve with the same start and end vertices and the minimum number of vertices with distance at most $\epsilon$ from the original curve.
However, for a set of curves, the simplification errors are aggregated, if they are computed independently.
We focus on the simplification of a set of curves by finding a representative curve that is a good cluster representative for $L_p$-based clusterings and has almost k vertices, assuming that the input curves have Fr\'echet distances at most $\epsilon$ from each other.

In a simplification algorithm, a \textit{shortcut} is a segment that replaces a part of the curve starting and ending at vertices of that curve.

The {\em Fr\'echet distance} is the minimum length of the leash between a man walking on one curve from the start to the end, and his dog walking on the other curve from the start to the end, given that none of them ever goes back.
Deciding the Fr\'echet distance between two curves takes $O(n^2)$ time for curves with $O(n)$ vertices using the free space diagram (FSD)~\cite{fsd}.
The Fr\'echet distance cannot be decided in $O(n^{2-\epsilon})$ time~\cite{seth} or even approximated by a factor better than $3$~\cite{3aprx}, for any $\epsilon > 0$, unless SETH fails.
For $L$ input curves with $O(n)$ vertices, assuming SETH is true, it is not possible to decide the Fr\'echet distance of the curves in $O(n^{L-\epsilon})$ time, for all $\epsilon>0$~\cite{kcurves}.

\subsection{Previous results.}

Computing a representative curve is a well studied problem~\cite{median,mid2,mid3,mid4,mid5,mid6,klcenter,klmedian}.
For similar (close) monotone trajectories with the same start and end vertices, Buchin, et al~\cite{median} presented algorithms for computing the median trajectory and the homotopic median trajectory.

A curve simplification where the points of the simplified curve should be a subset of the vertices of the input curve is called a {\em discrete curve simplification}.
Discrete curve simplification under Fr\'echet distance is solvable in $O(n^2)$ time~\cite{imaiiri}, and no algorithm with $O(n^{2-\epsilon})$ running time exists for all $\epsilon>0$, if SETH holds~\cite{kcurves}.

In the \textit{global min-$\epsilon$ simplification}, $P'$ is a subsequence of $P$ with at most $k$ vertices that minimizes $d_F(P',P)$.
The current best exact min-$\epsilon$ simplification algorithms for global Fr\'echet distance have cubic complexity~\cite{cubic}.

A similar problem is $k$-segment mean curve~\cite{feldman}, where a monotone path is simplified into a possibly discontinuous $k$-piecewise linear function. Also, the problem of min-$k$ simplification with arbitrary points of the plane, where the distance is given and the goal is to minimize the number of vertices has been discussed in~\cite{van2019global}.

Approximation algorithms with near linear time exist for {\em local simplification} under Fr\'echet distance~\cite{2k,abam}, where only the error of each shortcut is taken into account.
{\em Global discrete curve simplification} using Fr\'echet distance can be solved in $O(n^3)$ time~\cite{cubic}. If $\forall \forall \exists -OV$ conjecture holds, there is no algorithm for global simplification using Fr\'echet distance with running time $O(n^{3-\epsilon})$, for any $\epsilon>0$~\cite{cubic}.

The combination of the representative curve and curve simplification problems is the $(k,l)$-clustering problem, where the cost of clustering a set of curves $\{P_i\}_{i=1}^L$ into $k$ clusters with centers $\{C_i\}_{i=1}^k$, such that $$\mathbf{\sqrt[p]{\sum_{j=1}^n \min_i d(C_i,P_j)^p}}$$ is minimized using curves with complexity $\ell$ as cluster representatives (centers).

Driemel et al.~\cite{driemel} proposed $(1+\epsilon)$-approximation algorithms for $k$-center and $k$-median clustering of curves in 1D and a $2$-approximation for any dimensions, assuming the complexity of a center and $k$ is constant.
Buchin et al.~\cite{klcenter} presented an algorithm for computing the $k$-center of a set of curves under Fr\'echet distance, such that the complexity of the representative curves (centers) is fixed, and prove that it is NP-hard to find a polynomial approximation scheme (PTAS) for this problem. They also presented a $3$-approximation algorithm for this problem in the plane and a $6$-approximation for $d\geq 2$, and proved the lower bound $2.598$ for the discrete Fr\'echet distance in 2D if $P\neq NP$.

$p$-mean trajectories for $p\rightarrow\infty$ using $k$-center~\cite{klcenter,driemel}, and $p=1$ using $k$-median~\cite{driemel} exist.
The computation of $p$-MC based on the Fr\'echet reparameterization has already been discussed and implemented for $p=1,p=2$~\cite{buchin2019klcluster}, however, such a computation can have complexity $O(n^L)$, which is infeasible for large datasets. For $p=1$, the problem is W[1]-hard using $L$ as the parameter~\cite{buchin2020complexity}.

\subsection{Our results.}

We call the objective function of $p$-MC the \textit{$L_p$-norm of the Fr\'echet distance}.
Since the root function is monotone for $p\geq 1$, which are the values that appear in the cost of $L_p$-based clustering problems, it is sufficient to minimize the $p$-th power of Fr\'echet distance or $d_F(.,.)^p$.
Note that while both $L_p$-norms and the Fr\'echet distance satisfy the triangle inequality, their combination does not. For example, for $p\rightarrow \infty$, the inequality $\sqrt[p]{d(a,c)^p}\leq \sqrt[p]{d(a,b)^p+d(b,c)^p}$ becomes $$d(a,c) \leq \max (d(a,b),d(b,c)),$$ which does not always hold.

Given the re-parameterizations of the input curves that gives the optimal $p$-MC, the problem of finding the $p$-MC curve can be solved by reducing it to the point version of the problem, where a set of points is given and the goal is to find a point that minimizes the $\ell_p$-norm of distances from itself to the rest of the points. However, since the Fr\'echet distance only cares about the maximum distance between the points, only the points whose matching gives the maximum distance need to be considered. These are the Fr\'echet events.

Based on this observation, we give the following new results, and define new concepts that explain some of the reasons behind the good performances of existing algorithms and heuristics:
\begin{itemize}
\item We consider the $p$-MC for most values of $p$ and give approximation algorithms for them.
\Cref{table:results} summarizes the results on $p$-mean curves of $L$ curves.
\item We give a divide and conquer algorithm for computing the representative curve. The parallel implementation of our algorithm has time complexity independent from $L$.
\item We give a new simplification algorithm which can simplify an input curve with $O(\epsilon)$ error and $O(k)$ vertices, where $k$ is the length of the optimal simplification. It is based on a reduction to the discrete case.
\end{itemize}
\begin{table*}[ht]
\centering
\begin{tabular}{|p{2.6cm}|c|c|c|c|}
\hline
$p$-Mean & Time & $\epsilon$ & $k$ & Reference\\
\hline
Continuous $p$-Mean:&&&&\\
$p\rightarrow \infty$ & $\poly(n,k,L)$ & $\geq 2.25$ & $k$ & Lower bound~\cite{klcenter}\\
$p=1$ & $O(\poly(n,L))$ & $> 1$ & $n$ & Lower bound~\cite{klmedian}\\
$p\geq 1$ & $O(Ln^{5L}\log n)$ & $3$ & $k$ & \Cref{alg:continuous}\\
$p\geq 1$ & $O(L^2n^{2}\log n+\G(L,k))$ & $2$ & $2k$ & \Cref{alg:simp}\\
$p\geq 1$ & $O(L\T(n)+\G(L,k)))^{\dagger}$ & $2\alpha+1$ & $k$ & \Cref{alg:sensitive}\\
\hline
Discrete $p$-Mean:&&&&\\
$p\rightarrow \infty$ & $O(kLn\log n +n^2)$ & $6$ & $k$ &~\cite{klcenter}\\
$p\geq 1$ & $O(L^2n^2\log n+\G(L,k))^{\dagger}$ & $3$ & $k$ & \Cref{alg:pairwise}\\
\hline
\end{tabular}
\caption{Results on $p$-mean curve in $\mathbb{R}^2$.
$\T(n)$ is the complexity of an $\alpha$-approximation simplification algorithm and $\G(L,n)$ is the complexity of a $p$-MC algorithm with $k$ vertices for $L$ curves with $n$ vertices.
\\
The results marked with $^\dagger$ is for one recursion, they can be run on larger inputs by recursively calling themselves at the cost of increasing the approximation factor.}
\label{table:results}
\end{table*}
\section{Preliminaries}
A polygonal curve $P$ is a sequence of points $\{P[i]\}_{i=1}^n$ and the segments connecting each point $P[i]$ to its next point in the sequence, $P[i+1]$, for $i=1,\ldots,n-1$.

The \textit{Fr\'echet distance} of two curves $P,Q:[0,1]\rightarrow \mathbb{R}^2$ is defined as
\[
d_F(P,Q)=\inf_{\alpha,\beta} \max_{t\in[0,1]} d(P(\alpha(t)),Q(\beta(t))),
\]
where $\alpha$ and $\beta$ are reparameterizations, i.e., continuous, non-decreasing, bijections from [0,1] to [0,1], and $d(.,.)$ is a point metric.

In the Fr\'echet distance of a set of $L$ curves, $d$ is the diameter of the mapped points from the $L$ input curves and can be computed in $O(Ln^2 \log n)$ time~\cite{setofcurves}.
The Fr\'echet distance of $L$ curves is the diameter of their minimum enclosing ball or the $1$-center of the curves using Fr\'echet distance.
Using triangle inequality, the Fr\'echet distance of the curves is at most twice the distance from $1$-center to the farthest curve.

The \textit{free-space diagram} (FSD)~\cite{fsd} between two polygonal curves $P:[0,1]\rightarrow \mathbb{R}^2,Q:[0,1]\rightarrow \mathbb{R}^2$ for a constant error $\epsilon>0$, is a 2D region in the joint parameter space of those curves where each dimension is an arc-length parameterization of one of the curves, and the free space~(FS) is the set of all points that are within distance $\epsilon$ of each other:
$
D_{\epsilon}(P,Q)=\{(\alpha,\beta)\in[0,1]^2 \mid d(P(\alpha),Q(\beta))\leq \epsilon\},
$
and the rest of the points are non-free.
Therefore, each point of FSD defines a mapping/correspondence between a point on $P$ and a point on $Q$.
The Fr\'echet distance between two curves is at most $\epsilon$ iff there is an $\alpha\beta$-monotone path in the free space diagram from $(0,0)$ to $(1,1)$.
In figures, the free space is usually shown in white, and the non-free regions are shown in gray. The orthogonal lines drawn from the vertices of the input curves build a grid (FSD grid), whose cells are called the FSD cells.

A special transformed FSD called the deformed FSD was already defined for a variation of the Fr\'echet distance called the backward Fr\'echet distance, assuming the edges of input curves have weights~\cite{gheibi2019weighted}.

A curve $P$ is \textit{$c$-packed}~\cite{cpacked} if the total arc length of $P$ inside any ball of radius $r$ is at most $c\cdot r$.
The time complexity of computing a $(1+\epsilon)$-approximation of the Fr\'echet distance between $c$-packed curves is $O(\frac{cn}{\epsilon}+cn\log n)$~\cite{cpacked}.
$c$-Packed curves also have the property that for a given $\epsilon>0$, the complexity of the RS is within a constant factor of the complexity of the RS for $\alpha\cdot \epsilon$, for any $\alpha>0$.
The value $c$ of a $c$-packed curve can be approximated within factor $2+\epsilon$, for any $\epsilon>0$ in $O(\frac{\tau}{a})$ time, where $\tau$ is the length of the curve and $a$ is the distance between the closest points on the curve~\cite{aghamolaei2020windowing}.

$L_p$-based clustering problems are clusterings with cost equal to the $L_p$-norm of the distances between the points and their corresponding centers.
For a real number $p\geq 1$, the \textit{$L_p$ cost} of a set of points $T_i=(x_i,y_i)\in \mathbb{R}^2$, $i=1,\cdots,L$ is defined as
$$
\min_{T'\in \mathbb{R}^2} \sqrt[p]{\sum_{i=1}^L |T'-T_i|^p}.
$$
$T'$ is also the cluster center of $\{T_i\}_{i=1}^L$ in an $L_p$-based clustering.
There is a $(1+\epsilon)$-approximation algorithm using linear programming for computing the $L_p$ cost for fixed $p$~\cite{estimator}. For special cases such as $p=1,p=2$ and $p\rightarrow \infty$ explicit mathematical formulas for $x$ exist which can be used to compute $x$ in linear time.
Constant factor approximations for $L_p$-based clustering also exist~\cite{mappingcoreset}.

Given a curve $Q$ and a set of curves $\{P_i\}_{i=1}^L$, the \textit{$L_p$-norm of the Fr\'echet distances} is defined as
$$
\sqrt[p]{\sum_{i=1}^L d_F(Q,P_i)^p}.
$$
The $L_p$-norm of the Fr\'echet distances may not satisfy the triangle inequality for different curves $Q$.

Based on this definition, given the optimal mapping between the points of the curves, it is possible to compute the corresponding curve by finding the center of the mapped points for every pair of points from the curves by solving the $L_p$ cost optimization.

Finding the minimum-link (min-link) path in a polygonal domain asks for finding a minimum-link s-t path (a path from $s$ to $t$) such that the number of bends is minimized and the path lies inside the polygon and those not go through a set of polygonal holes. This problem can be solved in $O(E_{VG}\alpha^2(n)\log n)$ time, where $E_{VG}$ is the number of edges in the visibility graph of the polygon~\cite{mitchell1992minimum,toth2017handbook}.
\section{$p$-MC of two polylines}

\subsection{A certificate for $k$-simplification in FSD}
\paragraph{Certificate for Fr\'echet distance}
Given two curves $P,Q$ and a constant $\epsilon>0$, consider a set of certificates $C_{(i,j,p,q,p',q')}$ that indicate whether there is a $P$-monotone $Q$-monotone mapping of cost at most $\epsilon$ between the points of the $i$-th edge of the first curve $P$ with the $j$-th edge of the second curve $Q$, where $p$ is mapped to $q$ at the beginning of the mapping and $p'$ is mapped to $q'$ at the end of this mapping. Then, the Fr\'echet distance of $P$ and $Q$ is at most $\epsilon$ if a sequence of certificates $\pi_1,\ldots,\pi_k$ exists where the first points $(p,q)$ of the first certificate $\pi_1$ are the start vertices of the curves, namely $P_1,Q_1$, and the end vertices of the last certificate $\pi_k$ are the last vertices of the curves, namely $P_{|P|},Q_{|Q|}$.
Formally,
\begin{align*}
F_{(\epsilon,P,Q)}=\{&\exists \{\pi_i\}_{i=1}^k\in {[|P|]\times [|Q|] \times (P\times Q)^2}:\\
&q(\pi_i)=p(\pi_{i+1}), q'(\pi_i)=p'(\pi_{i+1}), \prod_{i=1}^k C_{\pi_i}=\text{true}\},
\end{align*}
where $[n]$ denotes the set $\{1,\ldots,n\}$, $f(t)$ denotes the member $f$ of the tuple $t$, and $C_{\pi_i}$ is the indicator variable for the Fr\'echet distance of two line segments.

\paragraph{Certificate for path of length $k$}
We define a certificate for a path between two points $p$ and $q$ of length $k$ with restriction on the feasibility of edges can be defined
by the recurrence relation
\begin{align*}
F(p,q,k)&=\vee_{t\in \mathbb{R}^2} (F(p,r,i)\wedge F(r,q,k-i)),
\\
F(p,q,1)&=\begin{cases}
\text{true} & \text{if there is a feasible edge $(p,q)$}\\
\text{false} & \text{otherwise}\\
\end{cases}.
\end{align*}

An example of this problem is the shortest path in a polygon with holes, where the feasibility constraint for the validity of a segment is that it does not intersect a hole.

\paragraph{Certificate for $k$-simplification under Fr\'echet distance}
Similarly, a certificate can be defined for the existence of a path of length $k$ and Fr\'echet distance at most $\epsilon$. Let $F_k$ be the certificate of existence of a path of length $k$. Then, the certificate for a simplification of length $k$ is given by $C_{\pi_i}\wedge F(p_1,p_n,k)$.
In \Cref{sec:nfsd}, we discuss how to build a diagram and define the certificate for $k$-simplification using the Fr\'echet distance on it.

Some previous work~\cite{cubic,van2019global} assume there is one interval on each edge and consider the first point of the interval~\cite{cubic}, or the whole interval~\cite{van2019global}. As shown in \Cref{fig:example}, this is not always the case.
\begin{figure}[h]
\centering
\includegraphics[scale=0.8]{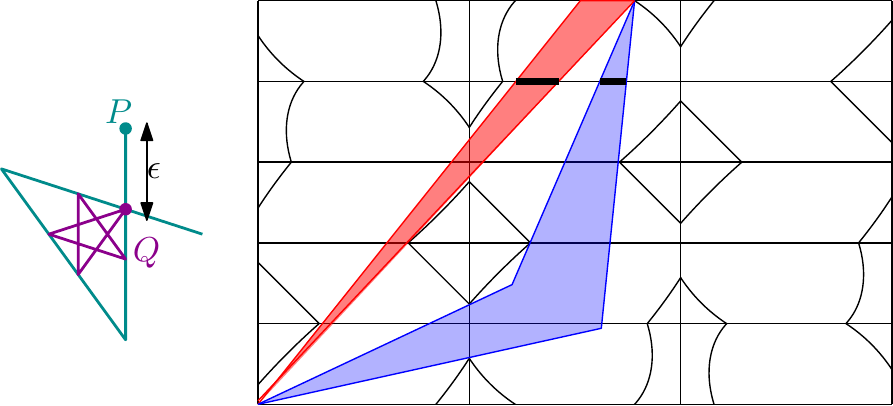}
\caption{Two intervals (bold lines) on the same edge for valid reparameterizations of the curves $P$ and $Q$ for $2$-simplification under Fr\'echet distance.}
\label{fig:example}
\end{figure}
\subsection{Normalized free-space diagram}\label{sec:nfsd}
Scaling the axes of the free-space diagram by a constant has already been discussed~\cite{gheibi2019weighted}.
In \Cref{lemma:transform} we show a more general transformation works.
\begin{lemma}\label{lemma:transform}
For a set of transformations $f_i:[0,1]\rightarrow [0,1]$, the free-space diagram with the free-space $\{(t_1,\ldots,t_L)\mid d(P_1(f_1(t_1)),P_2(f_2(t_2)),\ldots,P_L(f_L(t_L))) \leq \epsilon\}$ has the same set of feasible reparameterizations, if $f_i$ is a one-to-one non-decreasing function.
\end{lemma}
\begin{proof}
Substituting $t'_i=f_i(t_i)$ gives:
\[
(P_1(f_1(t_1)),P_2(f_2(t_2)),\ldots,P_L(f_L(t_L)))=(P_1(t'_1),P_2(t'_2),\ldots,P_L(t'_L)).
\]
The function is one-to-one, so $t_i=f^{-1}_i(t'_i)$. 
Since $t'_i:[0,1]\rightarrow[0,1]$ and its non-decreasing, it is still a reparameterization of $P_i$, and the set of reparameterizations remains the same.
\end{proof}

We introduce a scaled FSD called {\em normalized FSD} which changes the representation of the curves in FSD such that a path in FSD corresponds to a segment in the original space if the derivatives of any point on the curve with respect to each of the FSD axes (input curves) is the same.
We formalize this in \Cref{lemma:nfsd}, where the length of each segment from the input curves is divided by $\sqrt{m_{P_i}^2+1}$, and $m_{P_i}$ is the slope of the segment.
To handle negative slopes as well, we add at most $n$ points on each of the edges of FSD cells that correspond to the intersection of the extensions of the shortcuts through previous vertices with the corresponding segment of that FSD edge in the Euclidean plane.
This is formally explained in \Cref{lemma:negative}.
In the rest of the paper, when we use FSD, we mean the normalized free-space diagram.
\begin{lemma}\label{lemma:nfsd}
A segment in the normalized FSD corresponds to a segment in the original space (Euclidean plane), if the slopes of the edges of each input curve have the same sign, i.e. all positive slopes or all negative.
\end{lemma}
\begin{proof}
Assume $P$ and $Q$ are two input curves, and we want to find a condition on a curve $R$ in the parameter space of $P,Q$ that guarantees it will correspond to a polygonal curve in the Euclidean plane.

Choose an arbitrary segment from each of these curves. We want to change the mapping of the points of $P$ and $Q$ to the axes of the FSD to keep the slope of $R$ constant along different segments.
Let $\ell_P$ be the length of the curve from its start vertex to the point where the length of the curve $P$ reaches $\ell_P$. So, the domain of $\ell_P$ is $[0,\Len{P}]$, where $\Len{P}=\int_{0}^1 P(t) dt$ is the length of curve $P$. Similarly, define $\ell_Q$ and $\ell_R$.
\Cref{fig:units} shows unit vectors in the direction of $\ell_P,\ell_Q,\ell_R$ for a segment of $P$, respectively.

\begin{figure}[h]
\centering
\includegraphics[scale=0.5]{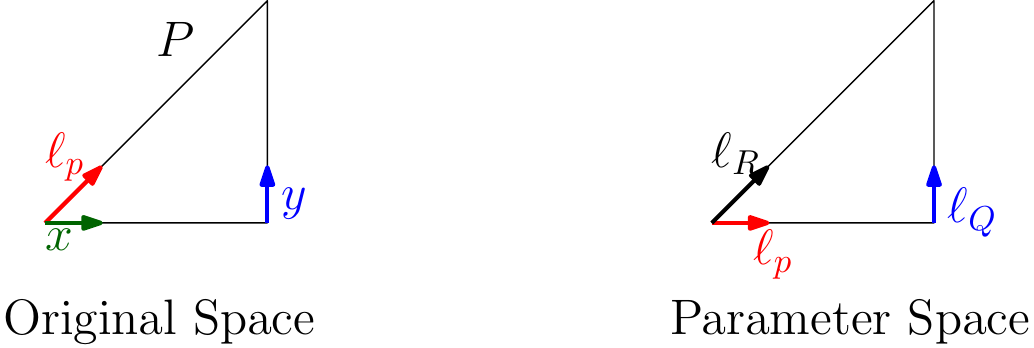}
\caption{Unit vectors in the Euclidean plane (original space) on the left, and in the FSD (parameter space) on the right.}\label{fig:units}
\end{figure}

For each segment of the curves, we define a reparameterization.
Let $P'(x)=m_P \cdot x+p_1,$ where $x\in [x_1,x_2]$, be a segment of curve $P$ with slope $m_P$ and $y$-intercept $p_1$.
The reformulation of $P'$ in terms of $\ell_P$ is given in the following formula:
$
P'(\ell_P)=p_1+\frac{m_P}{\sqrt{m_P^2+1}}\ell_P, \text{ where } \ell_P\in[0,\Len{P}],
$
since using the derivatives of length variables:
$
d \ell_P = \sqrt{(d x)^2+(d y)^2}, \quad \frac{d P'}{d x} = m_P \Leftrightarrow
$
\\
$
d \ell_P = \sqrt{(d x)^2+m_P^2 (d x)^2} = \sqrt{1+m_P^2} d x\Leftrightarrow
d P'(\ell_P) = \frac{d P'(x)}{d x} \frac{d x}{d \ell_P} = \frac{m_P}{\sqrt{m_P^2+1}}\ell_P.
$
Similarly, we reparameterize a segment $Q'$ of curve $Q$ in terms of its length variable $\ell_Q$.
The axes of the normalized FSD are $P$ and $Q$. So we need to compute the slope of the line segment $R'$ from curve $R$ in terms of $\ell_P$ and $\ell_Q$:
$
\frac{dR'}{d \ell_P} = \frac{dR'}{dx}\times \frac{dx}{d\ell_P} = \frac{dR'}{dx} \frac{1}{\sqrt{m_P^2+1}},
$
and the equation for $\ell_Q$ is similar.
This means that scaling each segment of the curve $P$ by a factor $\frac{1}{\sqrt{m_P^2+1}}$ preserves the slope of $R'$ in terms of $\ell'_P$ in the Euclidean plane.

Based on \Cref{lemma:transform}, if a transformation converts the ellipse that is the free-space inside a cell into a degenerate ellipse, it does not work anymore. We show how to handle these cases that they still preserves the properties. If the original free-space is a degenerate case, the transformation only changes the slope of the lines.

After scaling parts of the axes of the FSD, the slope of segment $R$ will not change if the slopes of the segments of the input curves have the same sign. If this is not the case, use another segment from one of the curves to compute the reparameterization, and then map the points accordingly. This is possible if at least one of the edges of one of the curves has a different slope; otherwise all the points on the each of the curves are collinear.
In that case, the original free-space is a degenerate ellipse, i.e. linear functions, for which the scaling by a constant factor as described in this lemma works.

\end{proof}
Without the transformation described in \Cref{lemma:nfsd}, a polyline in the FSD still represents a polyline in the Euclidean space, however, the number of vertices can be different.

To add the signs of the slopes of the shortcut segments, it is enough to add them to the boundaries of the cells, and instead of computing the minimum-link path, compute the unweighted shortest path with these vertices in addition to the intersections of the free-space with the FSD grid (cell boundaries). The edges have weight $0$ when the slope of the last segment is equal to the slope of the next segment, otherwise the edges have weight $1$. Because only non-decreasing paths from the start to the end vertex correspond to valid reparameterizations, direct the edges in the order of increasing index. By computing the shortest path in the resulting DAG, the minimum complexity polyline in NFSD is computed.

Let $S_{i,j}(P,Q)$ be the set of intersection points between the shortcuts $\overline{p_1p_t}$ for $t=1,\ldots,n$ and the segment $\overline{q_jq_{j+1}}$, for two polylines $P=(p_1,\ldots,p_n)$ and $Q=(q_1,\ldots,q_n)$. Build the graph $H=(V,E)$ with vertices $V=V_G \cup \cup_{i,j} S_{i,j}(P,Q)$, where $V_G$ is the set of vertices in FSD, i.e. the intersection points of the free-space with the grid lines of FSD, which are the points that map a vertex of one curve to each point on the other curve. The edges $E$ connect the vertices with a segment with non-negative slope between them that lies completely in the free-space, with an edge of weight $1$ in the graph.

Based on the definition, three cases for the simplification can be solved using $H$, depending on the subset we choose the vertices of the simplification from:
\begin{itemize}
\item Vertices of the input curves: use the vertex for the grid line containing the edge of the cell containing that point. Remember that each grid line is a vertex of an input polyline in all FSDs.
\item Vertices of one input curve (curve $P$): The shortest path must be computed using only the vertices of $V_G$ and the edges between the vertices of $V$ or with a subset of $\cup_{i,j} S_{i,j}$ only if the incoming and outgoing edges have the same slope after mapping to the Euclidean plane.
\item Any point of the Euclidean plane: It is enough to map the vertices of $V$ to the Euclidean plane.
\end{itemize}
\Cref{lemma:negative} shows the simplification can be computed using the shortest path on $H$, or a subgraph of it.
\begin{lemma}\label{lemma:negative}
The shortest path in $H$ from the start vertex to the end vertex gives the minimum complexity path in the Euclidean plane.
\end{lemma}
\begin{proof}
Based on \Cref{lemma:nfsd}, the edges of $E$ that connect two points where the slope of segments between them does not change, i.e. the part of the curve between them is monotone, map to a single segment in the Euclidean plane.

Consider a reparameterization that gives the optimal simplification. For each vertex that is shortcutted in this simplification, there is a point on the boundary of the cell intersecting that path in NFSD. Since $V=V_G \cup (\cup_{i,j} S_{i,j}(P,Q))$ contains all such intersections, its vertices are a subset of $V$. So, for every optimal simplification there is a path in NFSD.

To show every shortest path in NFSD gives an optimal simplification, for each vertex of the shortest path between two edges with different slopes in the Euclidean plane, choose a vertex from $P$ or $Q$ depending on the boundary edge that contained that point.

If we only want the vertices of $P$ to exist in the output solution, the shortest path must be computed using only the vertices of $V_G$ and the edges between the vertices of $V$ or with a subset of $\cup_{i,j} S_{i,j}$ only if the incoming and outgoing edges have the same slope after mapping to the Euclidean plane.

Based on the definition of the edges, the slope is preserved in each edge. So, each edge in NFSD is equivalent to a segment in the Euclidean plane. This means the weight of the curve in the NFSD is equal to the complexity of the polyline in the Euclidean plane.

\end{proof}

Adding $n$ points on each boundary edge increases the complexity of NFSD to $O(n^3)$. So, the simplification that minimizes the Fr\'echet distance between the simplified and original curves can be computed in $O(n^3 \log n)$ time and in $O(n^2\log n)$ time for monotone curves. For monotone curves, all the slopes have the same sign, so we do not need to add points on the boundaries of the cells.

Note that knowing only the slopes of the lines is not enough and the mapped length of the curve is also needed to define a $k$-simplification. More specifically, there can be $\binom{n}{k}$ partial $k$-simplifications of a single polyline that end at the same edge, which can result in distinct optimal matchings. In NFSD, this is equivalent to having multiple monotone shortest paths between $(0,0)$ and $(1,1)$. In a diagram with holes, the intervals on the edges that represent these partial solutions might not be continuous. In NFSDs/FSDs, the free-space acts as a certificate for valid partial matchings for Fr\'echet distance, however, the certificates for $k$-simplifications are only covered by NFSDs.
\subsection{Exact $p$-mean curve}
In \Cref{lemma:mapping}, we show the reparameterization that gives the $p$-mean curve of two curves is the one that gives the Fr\'echet distance between them.

\begin{lemma}\label{lemma:mapping}
The Fr\'echet reparameterization of curves $P$ and $Q$, minimizes the distances to the $p$-MC of $P$ and $Q$.
\end{lemma}
\begin{proof}
Let $s\in P$ and $q\in Q$ be a pair of points mapped to each other in the Fr\'echet mapping between $P$ and $Q$. Let $m$ be the point on $\vv{sq}$ which lies on the $p$-MC of $P,Q$. The goal is to minimize the cost of $p$-MC for these points:
$
\min_m |sm|^p+|mq|^p  = \min_m |sm|^p+(|sq|-|sm|)^p.
$
Then, we take the derivative of the above cost expression $|sm|^p+(|sq|-|sm|)^p$ in terms of $|sm|$:
\\
$
p|sm|^{p-1}-p(|sq|-|sm|)^{p-1}=0 \Leftrightarrow
|sm|^{p-1}=(|sq|-|sm|)^{p-1}\Rightarrow
|sm|=|sq|/2.
$
This is a minimum of the function, since for $|sm|>|sq|/2$ the derivative is positive and for smaller values it is negative.
Substituting this value in the cost expression gives $2^{1-p}|sq|^p$.
This means that the minimum of $|sq|$ also minimizes the cost expression.
The maximum of $|sq|$ for all pairs $s,q$ is the maximum distance in the reparameterizations of $P,Q$ that realizes the Fr\'echet distance.

\end{proof}

Based on \cite{setofcurves}, the higher dimensional FSDs (for more than two curves) can be constructed by building the FSDs of pairs of curves, extending them in the direction of the axes corresponding to the rest of the curves, and taking their intersection. Based on \Cref{lemma:transform}, NFSD represents the same set of reparameterizations as FSD.
A $(\epsilon_1,\ldots,\epsilon_L)$-NFSD is the $L$-dimensional NFSD in which the NFSD for each pair $P_i,P_j$ of the input curves uses distance $\epsilon_i$ from $P_i$ and distance $\epsilon_j$ from $P_j$ as the distance (to define the free-space).

\begin{lemma}\label{lemma:epsilons}
Given a set of non-negative constants $\epsilon_1,\ldots,\epsilon_L$ and a set of curves $P_1,\ldots,P_L$, the shortest path in the $\overrightarrow{\epsilon}$-NFSD of $P_1,\ldots,P_L$ gives the minimum-link path in the Euclidean plane with distance at most $\epsilon_i$ from $P_i$, for each $i=1,\ldots,L$. Assume $\overrightarrow{\epsilon}=(\epsilon_1,\ldots,\epsilon_L)$ and $\lVert \overrightarrow{\epsilon} \rVert_p=\eta$.
\end{lemma}
\begin{proof}
A point $p$ in the free-space inside each cell of an $\overrightarrow{\epsilon}$-NFSD for $L$ curves satisfies:
\[
\forall i=1,\ldots,L \;:\; \lVert P_i(t_i)- p\rVert_2\leq \epsilon_i,
\]
So,
\[
\lVert \overrightarrow{\epsilon} \rVert_p=\eta \Leftrightarrow\sum_{i=1}^n \epsilon_i^p = \eta^p \Leftrightarrow \sum_{i=1}^n \lVert P_i(t_i)- p\rVert_2^p \leq \eta^p.
\]
These are a set of $L$ ellipses, which are monotone except at their extreme points and the boundaries of the domain of their definition.
Candidates for the optimal matchings of each point are the intersections between the grid lines, the extentions of the shortcut lines, and the ellipses.

The free-space inside each cell of NFSD for two curves is an scaled ellipse, as proved in \Cref{lemma:nfsd}; the higher dimensional NFSD is similarly proved to have an ellipsoid inside each cell as the free-space.

Using \Cref{lemma:negative}, the complexity (the number of vertices) of the shortest path in NFSD is equal to the complexity of the simplification in the original space.
Since the scaling constants in each dimension are independent from each other, \Cref{lemma:nfsd} generalizes to any dimension, i.e. any number of curves.

Let $M$ be the minimum-link path from $(P_1(0),\ldots,P_L(0))$ to $(P_1(1),\ldots,P_L(1))$ in this $\overrightarrow{\epsilon}$-NFSD.
We showed that $M$ satisfies the Fr\'echet distance $\epsilon_i$, i.e. for each point $m\in M$, the distances to each of the curves satisfy $d(m,P_i(t_i))\leq \epsilon_i$.
Any optimal simplification maps to a path $M'$ in $\overrightarrow{\epsilon}$-NFSD, based on~\Cref{lemma:transform}. $M'$ is a polyline for because $\overrightarrow{\epsilon}$-NFSD preserves the slope of the lines with respect to the input curves and changing $p$ only effects the shape of the free-space.

\end{proof}

The changes to the graph $H$ built from an $\overrightarrow{\epsilon}$-NFSD after changing $\overrightarrow{\epsilon}$ form a discrete set of events, i.e. values $\overrightarrow{\epsilon}$ at which $H$ changes.
In \Cref{lemma:local}, we discuss the events at which the complexity of the shortest path in $\overrightarrow{\epsilon}$-NFSD change, i.e. the certificates for $p$-MC.
\begin{lemma}\label{lemma:local}\label{lemma:estimation}
The number of events for the min-$\epsilon$ $p$-MC simplification of a curve $P$ with respect to $L$ curves is $O(n^{2L})$, and each event can be computed in $O(1)$ time.
\end{lemma}
\begin{proof}
Changes to the graph $H$ built from an $\overrightarrow{\epsilon}$-NFSD when changing $\overrightarrow{\epsilon}$ happen at the intersections of the free-space inside the cell with the cell boundary, or when the monotonicity of the path between the cells changes which happen when the intersections of intervals on the boundaries of the cells. In \Cref{lemma:epsilons}, we showed the shape of the free-space inside a cell is a transformed unit ball of $L_p$ norm, i.e. $x^p+y^p=1$. For each edge of one curve and a vertex from the rest of the polylines, the intersection of this transformed ball and the boundary gives $L-1$ intervals.
These events are the intersections of the transformed ball with the boundary, and the intersections of the projections of the intervals on their shared edge.
Changing $\overrightarrow{\epsilon}$ scale and translates the transformed unit ball of the $L_p$ norm that represents the free-space inside each cell. So, the intersections of it with each cell boundary is still one continuous interval for each (vertex,edge) pair.
Each NFSD has $O(n^L)$ cells, each with $2L$ intersections, so the number of these events is $O(Ln^L)$.

For different slope signs, instead of a straight line segment, we look for segments $s,s'$ that share an endpoint on the edge $e$ between the cell with different slopes and its neighboring cells, such that $s$ and $s'$ are the reflections of each other with respect to $e$. As discussed in the definition of $H$ for NFSD, these are the set of shortcuts and their extensions or reflections in case of slope changes.
Based on the type of simplification, the size of $H$ is different:
\begin{itemize}
\item Any point on one of the input curves can be used in $p$-MC:\\
There are $\lvert P_i \rvert \prod_{j=1, j\neq i}^L \lvert P_j\rvert=O(n^L)$ shortcuts, each intersecting with each of the $O(n^L)$ edges of the NFSD grid, resulting in $O(n^{2L})$ event points.
\item The vertices of one of the input curves, $P_i$:\\
Only the edges of $H$ that change slope at a point on the edges that are on the grid line for a vertex of $P_i$ are allowed.
Since $O(n^{L-1})$ edges remain, each containing $O(n^{L-1})$ points on them ($S_{i,j}$), the number of events is $O(n^{2L-2})$.
\end{itemize}

\end{proof}

Based on the generalization of NFSD in \Cref{lemma:epsilons}, we solve the $p$-MC problem in \Cref{theorem:pmc}.
\begin{theorem}\label{theorem:pmc}\label{lemma:time}
The $p$-MC of a $L$ curves can be computed in $O(Ln^{5L}\log n)$ time.
\end{theorem}
\begin{proof}
As long as the graph built on $\overrightarrow{\epsilon}$-NFSD is not changed, changing $\overrightarrow{\epsilon}$ will not change the solution. So, it is enough to check the values $\overrightarrow{\epsilon}$ from \Cref{lemma:local}. For each of these values, we build a $\overrightarrow{\epsilon}$-NFSD and compute the shortest path in $H$, one of which is the $p$-MC of the curves (\Cref{lemma:negative}). There are $O(n^{2L})$ values $\overrightarrow{\epsilon}$, and computing the shortest path in an $\overrightarrow{\epsilon}$-NFSD takes $O(Ln^{3L}\log n)$ time, resulting in a $O(Ln^{5L}\log n)$ time algorithm.

\end{proof}
\Cref{alg:continuous} implements \Cref{theorem:pmc}. Also, the algorithm can be used to compute a min-$\epsilon$ simplification.

\begin{algorithm}[h]
\caption{Continuous $p$-Mean Curve}
\label{alg:continuous}
\begin{algorithmic}[1]
\Require{Trajectories $P_1,\ldots,P_L$, an integer $p\geq 1$, a constant $\epsilon>0$}
\Ensure{A trajectory $M$}
\State{$\Epsilon$= The values $\epsilon$ for the events of $p$-MC.}
\State{$\overrightarrow{\epsilon_0}=\overrightarrow 0$}
\For{$k=0,\ldots,n$, and $\overrightarrow{\epsilon_0}=\overrightarrow 0$}
\For{$\overrightarrow{\epsilon} \in \Epsilon$}
\State{Build a $\overrightarrow{\epsilon}$-NFSD $F$ for $\{P_i\}_{i=1}^L$.}
\State{$\tau_{\overrightarrow{\epsilon}}=$ Find the shortest path in $H$ built from $F$.}
\State{$\overrightarrow{\epsilon_0}$=the $\overrightarrow{\epsilon}$ with $\lVert\overrightarrow{\epsilon}\rVert_p\leq \epsilon$ for which a path of length $k$ is found.}
\EndFor
\EndFor
\State{Build a curve $M$ by reporting the points of $\tau_{\overrightarrow{\epsilon_0}}$ in the Euclidean plane.}
\end{algorithmic}
\end{algorithm}

Since the time complexity of the exact $p$-MC algorithm is exponential in the number of curves $(L)$, we discuss approximation algorithms for such cases.

\subsection{Reducing the Continuous Version to the Discrete Version: A Simplified Version of \Cref{alg:continuous}}
Here, we discuss a simplified version of \Cref{alg:continuous} for $p=\infty$ and $L=2$, where instead of computing the path in the parameter space (FSD), we simulate the algorithm in the original space while computing the dynamic program for FSD. So, the complexity does not depend on the ply.

Let $H$ be the polygon built on the intersections of the ellipses (the free spaces) in each cell $C$ and the boundary of the cell. For each pair of vertices $p,q$ in $H$, compute the intersection of $pq$ with the boundary of $H$, including the boundary of the holes. Also, consider the extension of the shortcuts and the intersections between them.
Add the points on the original curves corresponding to these points in FSD.
Note that constructing the free-space diagram is not necessary to compute these events.

In \Cref{alg:simp}, we used unit direction vectors to indicate the slopes of the segments.
The output of the algorithm is also curve-restricted, i.e. the vertices of the simplification lie on the edges of the input curve. \Cref{lemma:simplify} formulates the effects of the modifications.
\begin{algorithm}[h]
\caption{Bi-criteria approximation of $p$-MC of two curves}\label{alg:simp}
\begin{algorithmic}[1]
\Require{A free-space diagram $D$ of $P$ with error $\epsilon$}
\Ensure{A monotone path $T$ from $P[0]$ to $P[n-1]$}
\State{Build $G$ from $D$.}
\State{$Q=$ Add the events of $G$ to $P$ with error $\epsilon$.}
\State{$D'=$ the FSD for $Q$ with itself.}
\State{$S[i][j]=\infty, \rho[i][j]=-1\quad\forall i,j\in[1,|Q|]$}
\State{$S[1][1]=0$}
\For{$i=2,\cdots,|Q|$}
\For{$j=1,\cdots,|Q|$}
\If{$F(D',Q[i-1,\cdots,i],Q[k,\cdots,j])$}
\State{$S[i][j],\rho[i][j]=\min_{k=0,\cdots,j} S[i][k]+
\begin{cases}
0 &\text{if $1_{\overrightarrow{Q[i-1]Q[i]}}=1_{\overrightarrow{Q[\rho[i][j]]Q[j]}}$}\\
1 &\text{otherwise}
\end{cases}
,k
$}
\Comment{If the last edge and the next edge are collinear, do not increase the length, otherwise, increase it by one.}
\EndIf
\EndFor
\EndFor
\If{$S[i][j]=-1$}
\\ \Return{FAILED}
\EndIf
\State{$i\gets|Q|,j\gets|Q|$}
\While{$i\geq1 \wedge j\geq1$}
\State{Add $S[i]$ to the end of $T$, if $1_{\overrightarrow{T[i-1]T[i]}}\neq 1_{\overrightarrow{T[i]Q[\rho[i][j]]}}$.}
\State{$i\gets i-1, j\gets \rho[i][j]$}
\EndWhile
\\ \Return{$T$ in the reverse order}
\end{algorithmic}
\end{algorithm}
An intersection between two shortcuts is a point in FSD that does not fall on an edge of FSD grid, i.e. it does not map a vertex of one of the input curves, therefore, it cannot be chosen as a vertex of the simplification. \Cref{lemma:simplify} proves there is a path of twice the length that goes uses only the vertices of one of the curves.
\begin{lemma}\label{lemma:simplify}
Each monotone path of length $k$ in the parameter space can be mapped to a monotone path of length at most $2k$ in $H$, if the events of the intersections between the shortcuts are not used as the vertices of the simplification.
\end{lemma}
\begin{proof}
Consider a monotone path in the FSD with a point $p_i$ on the monotone path inside the free space. Let $C$ be the cell containing $p_i$.
Compute the intersection of the neighboring edges $p_{i-1}p_i$ and $p_ip_{i+1}$ with $H$ and call them $u$ and $v$ (See \Cref{fig:2k}).
\begin{figure}[h]
\centering
\includegraphics[scale=0.3]{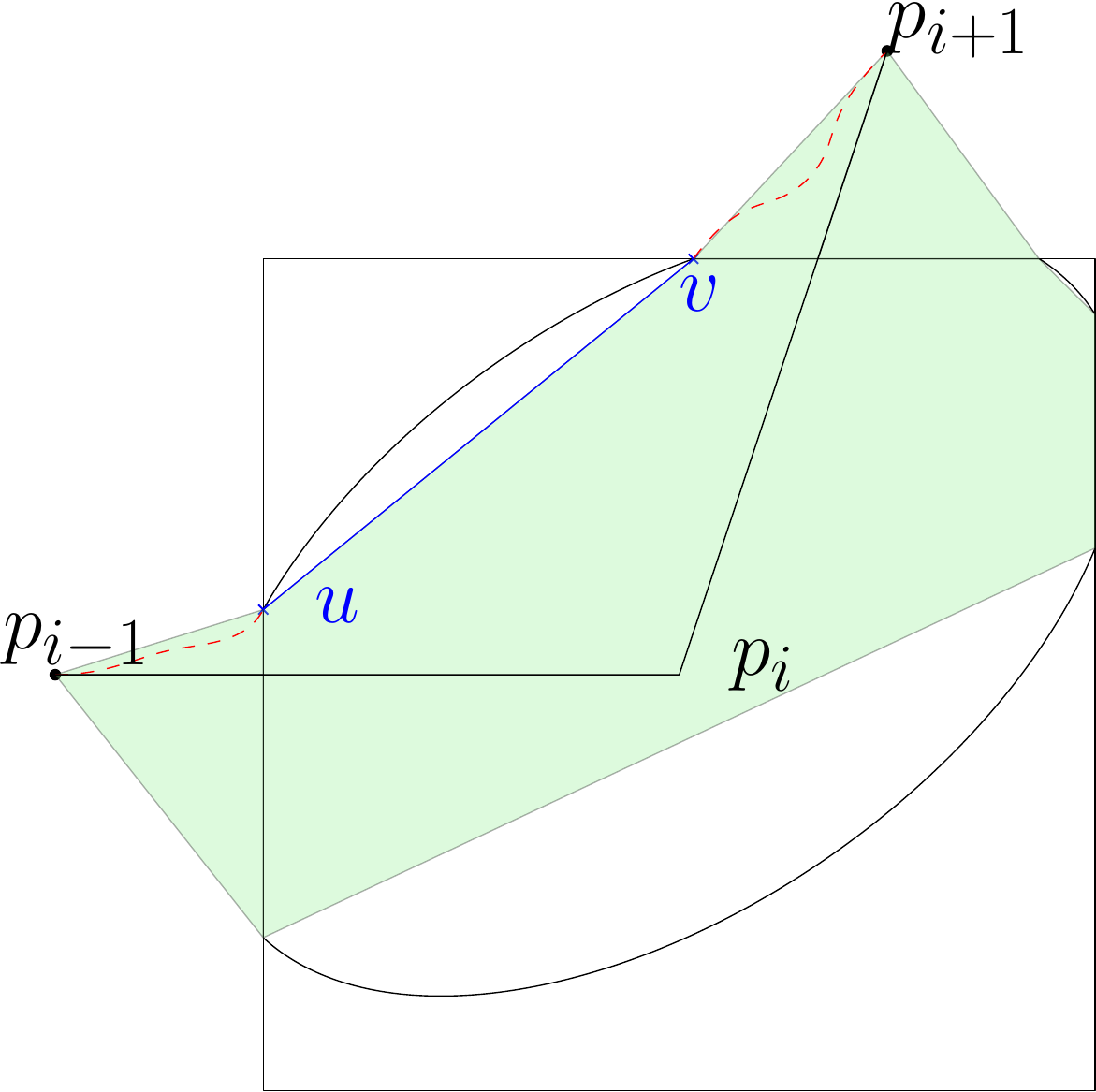}
\label{fig:2k}
\end{figure}
Since the free space in each cell is an ellipse, and therefore convex, replacing $p_i$ with $u$ and then $v$ does not change the monotonicity of the curve, and only increases the length of the curve by $1$. Using the same convexity argument in neighboring cells, $p_{i-1}u$ and $vp_{i+1}$ fall inside the free space.
This can happen once for each vertex, since $p_i$ will be on an edge of the grid (equivalently a vertex of the curve) after that.
By induction on the length of the curve, repeating this for all vertices gives a path of length at most $2k$.

\end{proof}
\Cref{alg:simp} adds a vertex for the intersections of disks of radius $\epsilon$ with every edge and shortcut, and then computes the simplification. In \Cref{theorem:epsilons}, we show there is a $(2k)$-simplification of error at most $2\epsilon$, if a $k$-simplification with error at most $\epsilon$ exists.
\begin{theorem}\label{theorem:epsilons}
\Cref{alg:simp} finds a simplification of the input curve with at most $2k$ vertices and error at most $2\epsilon$, where $k$ is the size of the optimal simplification using any subset of points in the plane as vertices.
\end{theorem}
\begin{proof}
The vertices of the input curve that replace the vertices of the optimal simplification with any subset of the points in the plane as vertices, as described in \Cref{lemma:simplify}, can be replaced by two points on the curve and with distance at most $\epsilon$ from the simplification (points $u$ and $v$ from \Cref{lemma:simplify}). This is because the free space interval on each edge has distance at most $\epsilon$ from one of the curves.
If a curve simplification using points of the plane with Fr\'echet distance $\epsilon$ (or equivalently, a re-parameterization of the curves with distance $2\epsilon$) exists, each part of the curve that lies inside a disk of radius $\epsilon$ can be covered by the center of that disk (a vertex of the input curve), so it does not increase the Fr\'echet distance. Since the algorithm restricts the points to be on the curves, any point inside the disk can be chosen as part of the output instead of its center, resulting in estimation error $2\epsilon$ which is the sum of the errors at each endpoint of an edge of the optimal simplification. So, the previous Fr\'echet mapping (between the optimal simplification and the input) can be used with Fr\'echet distance $2\epsilon$ for the $(2k)$-simplification, because the points of the $(2k)$-simplification have distance at most $\epsilon$ to the points of the optimal $k$-simplification, and the optimal $k$-simplification has distance at most $\epsilon$ to the input curve.

The complexity of the curve follows from \Cref{lemma:simplify}. This argument in the Euclidean plane is equivalent to that a point on the part of an edge that lies inside a disk might be replaced by the intersection points with that disk, and therefore double the complexity of the computed path.

\end{proof}
\section{$p$-Mean Curve of A Set of Curves}
By substituting the Fr\'echet events of two curves with the Fr\'echet events of $L$ curves, the results of the previous section extend to $L$ curves, for $p\rightarrow \infty$, since the maximum of the distances is considered. For other $p$-mean curves, their distances to the $p$-MC can be different, so the previous methods do not apply.

In this section, we discuss two algorithms for $p$-MC of $L$ curves and analyze their approximation factors.
\subsection{The Pairwise Algorithm for Discrete $p$-Mean Curve}
The $p$-MCs of a set of curves, like simplification using the Fr\'echet distance, is not unique. So, dividing the computation of the $p$-MC with at most $k$ vertices into first computing the Fr\'echet distance of a set of curves, and then simplifying the resulting curve does not yield the optimal solution.

\Cref{alg:pairwise} computes an approximate $p$-MC. In this algorithm, all the Fr\'echet distances between the curves are computed, then, the simplification of the one with the minimum distance to the rest of the curves is reported as an approximate solution.
\begin{algorithm}[h]
\caption{Pairwise Algorithm}
\label{alg:pairwise}
\begin{algorithmic}[1]
\Require{A set of trajectories $\{P_i[1..n]\}_{i=1}^L$}
\Ensure{A $p$-mean trajectory $M$}
\State{$D[i,j]=d_F(P_i,P_j), \forall i=1,\ldots,L,j=1,\ldots,L$}
\State{$M = \arg \min_{i=1}^L \sqrt[p]{\sum_{j=1}^L D[i,j]^p}$}
\State{$M \gets $ min-$\epsilon$ simplification of $M$.}
\end{algorithmic}
\end{algorithm}

\begin{lemma}\label{lemma:aprx1}
\Cref{alg:pairwise} is a $3$-approximation for discrete $p$-MC.
\end{lemma}
\begin{proof}
Assume $P_i$ is the curve that has the optimal solution $O_i$ as its simplification and let $T_i$ be an optimal min-$\epsilon$ simplification of $T_i$. Since $T_i$ is an optimal simplification, then
$
d_F(T_i,P_i) \leq d_F(O_i,P_i).
$
Since $p$-MC is also a simplification for $P_i$, its distance to $P_i$ is at least as much as the optimal simplification. Using triangle inequality of norms, the approximation factor is proved:
\begin{align*}
&\norm{d_F(P_j,T_i)}\\
&\quad\leq \norm{d_F(P_j,O_i)+d_F(T_i,O_i)}\\
&\quad\leq \norm{d_F(P_j,O_i)}+\norm{d_F(T_i,O_i)}\\
&\quad\leq \norm{d_F(P_j,O_i)}+\norm{d_F(T_i,P_i)}+\norm{d_F(O_i,P_i)}\\
&\quad\leq \norm{d_F(P_j,O_i)}+2\norm{d_F(O_i,P_i)} \leq 3 OPT.\*
\end{align*}

\end{proof}

\begin{lemma}\label{lemma:pairwise}
The time complexity of \Cref{alg:pairwise} is $$O(L^2 n^2\log n+\G(L,k))$$ for discrete simplification, using a $p$-MC algorithm of time $\G(L,k)$.
\end{lemma}
\begin{proof}
Computing the Fr\'echet distance of two curves takes $O(n^2\log n)$ time. Testing each curve $P_i$ as the center and computing the $L_p$ norm of the Fr\'echet distance of all curves $\{P_i\}_{i=1}^L$ requires $\binom{L}{2}$ distance computations between each pair of curves. This takes $O(L^2 n^2\log n)$ time. Finding the minimum takes $O(L^2)$ time.
Computing the $p$-MC with $k$ vertices takes $\G(L,k)$ time.

\end{proof}
Note that in \Cref{alg:pairwise}, while distances in matrix $D$ satisfy the triangle inequality, their $p$-th power does not. So, approximation algorithms based on triangle inequality cannot be used to prune away large distances in $D$.
\subsection{An Algorithm for $p$-Mean of $L$ Curves}
\Cref{alg:sensitive} simplifies the input curves with error less than their distances to the optimal $p$-MC, then it computes an approximate $p$-MC.

\begin{algorithm}[h]
\caption{$p$-Mean Algorithm}
\label{alg:sensitive}
\begin{algorithmic}[1]
\Require{A set of curves $\{P_i\}_ {i=1}^L$, an integer $k$, an $\alpha$-approx. min-$\epsilon$ simplification algorithm}
\Ensure{An approximate $p$-mean curve $M'$}
\For{$i=1,\ldots,L$}
\State{$P'_i=$ an approximate min-$\epsilon$ simplification of $P_i$.}
\EndFor
\State{$M'=$ a $p$-MC of $P'_1,\ldots,P'_L$.}
\end{algorithmic}
\end{algorithm}

\begin{theorem}\label{lemma:approx}
The approximation factor of \Cref{alg:sensitive} is $2\alpha+1$, if an $\alpha$-approximation simplification algorithm is used.
\end{theorem}
\begin{proof}
$M'$ denotes the $p$-mean of curves $\{P_i\}_ {i=1}^L$ computed by the algorithm, and $M$ denotes the optimal solution.
$P'_i$ is a simplification with error equal to the minimum error of simplifications of $P_i$ with at most $k$ vertices, so it has a distance less than any other curve, including $M$:
$
d_F(P_i,M) \geq d_F(P_i,P'_i).
$
Since $M'$ is the $p$-MC with minimum cost for curves $P'_i$, it has a lower cost than $M$. Using triangle inequality of Fr\'echet distance:
\begin{align*}
&d_F(P'_i,M')\\
&\quad\leq d_F(P'_i,M) \leq d_F(P'_i,P_i)+d_F(P_i,M)\\
&\quad\leq \sqrt[p]{d_F(P'_i,P_i)^p+d_F(P'_i,M)^p}\\
&\quad\leq \sqrt[p]{d_F(P'_i,P_i)^p+(d_F(P_i,P'_i)+d_F(P_i,M))^p}
\end{align*}
\begin{align*}
&\frac{d_F(P_i,M')}{d_F(P_i,M)}\\
&\quad\leq \sqrt[p]{(\frac{d_F(P'_i,P_i)}{d_F(P_i,M)})^p+(\frac{d_F(P_i,P'_i)+d_F(P_i,M)}{d_F(P_i,M)})^p} \\
&\quad\leq \sqrt[p]{1+2^p}.
\end{align*}
Substituting the approximation factors for computing $P'_i$ from $P_i$ gives the approximation factor:
$$
\sqrt[p]{\alpha^p+(\alpha+1)^p} \leq 2\alpha+1.
$$

\end{proof}

\begin{theorem}\label{theorem:time2}
\Cref{alg:sensitive} takes $O(L\T(n)+\G(L,k))$ time for continuous $p$-MC, if a $\G(L,n)$ time $p$-MC algorithm on $L$ curves, each with complexity at most $n$, is used.
\end{theorem}
\begin{proof}
Computing the simplification of a set of curves takes $\T(n)$ time.
The simplification algorithm is used $L$ times in the first step of the algorithm, so the total time complexity of that step is $L\T(n)$.
Computing the $p$-mean of curves $P'_i, i=1,\ldots,L$ takes $\G(L,k)$ time.
So, the running time of the algorithm is $O(L\T(n)+\G(L,k))$.

\end{proof}
\section{Experiments}
In this section, we use two types of trajectory data to evaluate our algorithm. The first one is a set of GPS tracks in different cities, and the second one is a set of pen trajectories while writing characters on a tablet. Our divide and conquer method when used in combination with our simplification algorithm produces results with good approximation factor, faster than existing algorithms that have at least quadratic time complexity due to computing the Fr\'echet distance on the whole input.
\subsection{GPS Trajectory Datasets}
We used two tracks from the datasets of \cite{ahmed2015comparison} of map construction repository~\cite{map}, which are GPS coordinates of trajectories in several cities.
One of them is track 29 of Athens Small dataset with 47 points. The other one is track 82 from Chicago dataset with 363 points.
In this experiment, we only consider the first two coordinates of the tracks and use $\epsilon$ as the simplification error and $\delta$ as the rounding error, and for the simplification algorithm we use \Cref{alg:simp}.

In \Cref{fig:athen}, the original curve and its simplification with $\epsilon=100, \delta=1$ are given. The number of events used is $128$ and the output size is $14$.
\begin{figure}[h]
\centering
\includegraphics[scale=0.5,trim={8.8cm 2cm 8.8cm 2cm},clip]{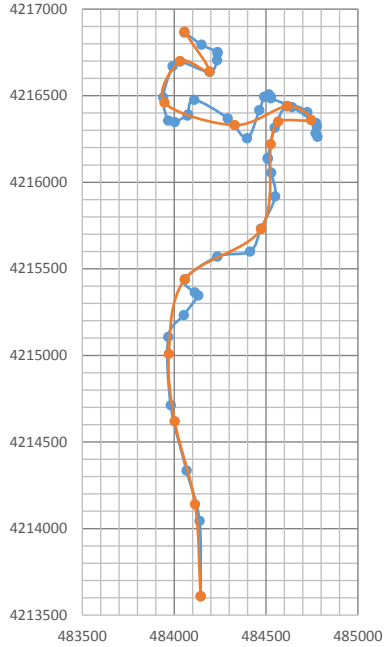}
\caption{Track 29 of Athens Small dataset (in blue) and its simplification (in orange).}
\label{fig:athen}
\end{figure}

Since track 82 of Chicago dataset is too large for the algorithm to compute fast, we break it into chunks of $30$ points, which gives $8$ chunks. Also, we first compute a simplification and then a simplification using the points of the plane on each part of the curve, and concatenate the results. The output size is $121$ points. The parameters are $\epsilon=20$ and $\delta=2$.
\begin{figure*}[htb]
\centering
\includegraphics[scale=0.6,trim={2cm 7cm 2cm 7cm},clip]{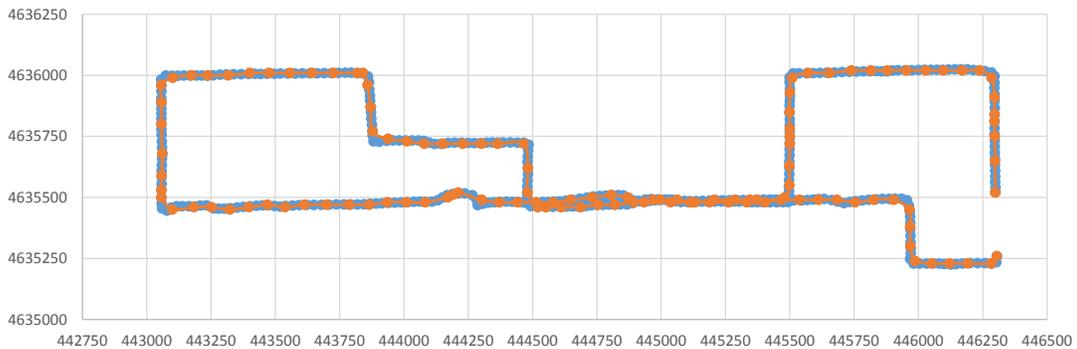}
\caption{Track 82 of Chicago dataset (in blue) and its simplification (in orange).}
\label{fig:chicago}
\end{figure*}
\subsection{Character Trajectories Dataset}
We ran the algorithm on the first trajectory of the dataset ``Character Trajectories Data Set''~\cite{williams2006extracting,williams2007primitive,williams2007modelling} from UCI Machine Learning Repository~\cite{ml}. Data is the pen tip trajectories of characters with a single pen-down which was captured using a WACOM tablet with sampling frequency 200Hz, and the data was normalized and smoothed.
The dataset contains 2858 character samples.

We only considered the first two dimensions x and y, and ignored the last dimension which was the pen-tip force. Since the points are close to each other, the number of events can be high, so, we partition each trajectory into subsets of smaller size, compute their simplification (using the points of the plane), and then attach them and compute their overall estimation.

The first trajectory has 178 points, and the input was partitioned into chunks of $30$ points. In \Cref{fig:char1}, its simplification (using the points of the plane) with error $\epsilon=0.1$ and the rounding error $\delta=0.01$ is shown, where the size of the output is $20$. In \Cref{fig:parts}, the output for $\epsilon=0.05$ and $\delta=0.001$ is shown, and the output size is $35$. Running the algorithm on the concatenation of the estimations of the parts with the same error is shown in \Cref{fig:merged} and it gives an output of size $31$, and the overall error is the sum of the errors, which is $\epsilon=0.1,\delta=0.002$.
\begin{figure}[h]
\begin{subfigure}{.45\textwidth}
\centering
\includegraphics[scale=0.25,trim={7cm 2.5cm 7cm 4cm},clip]{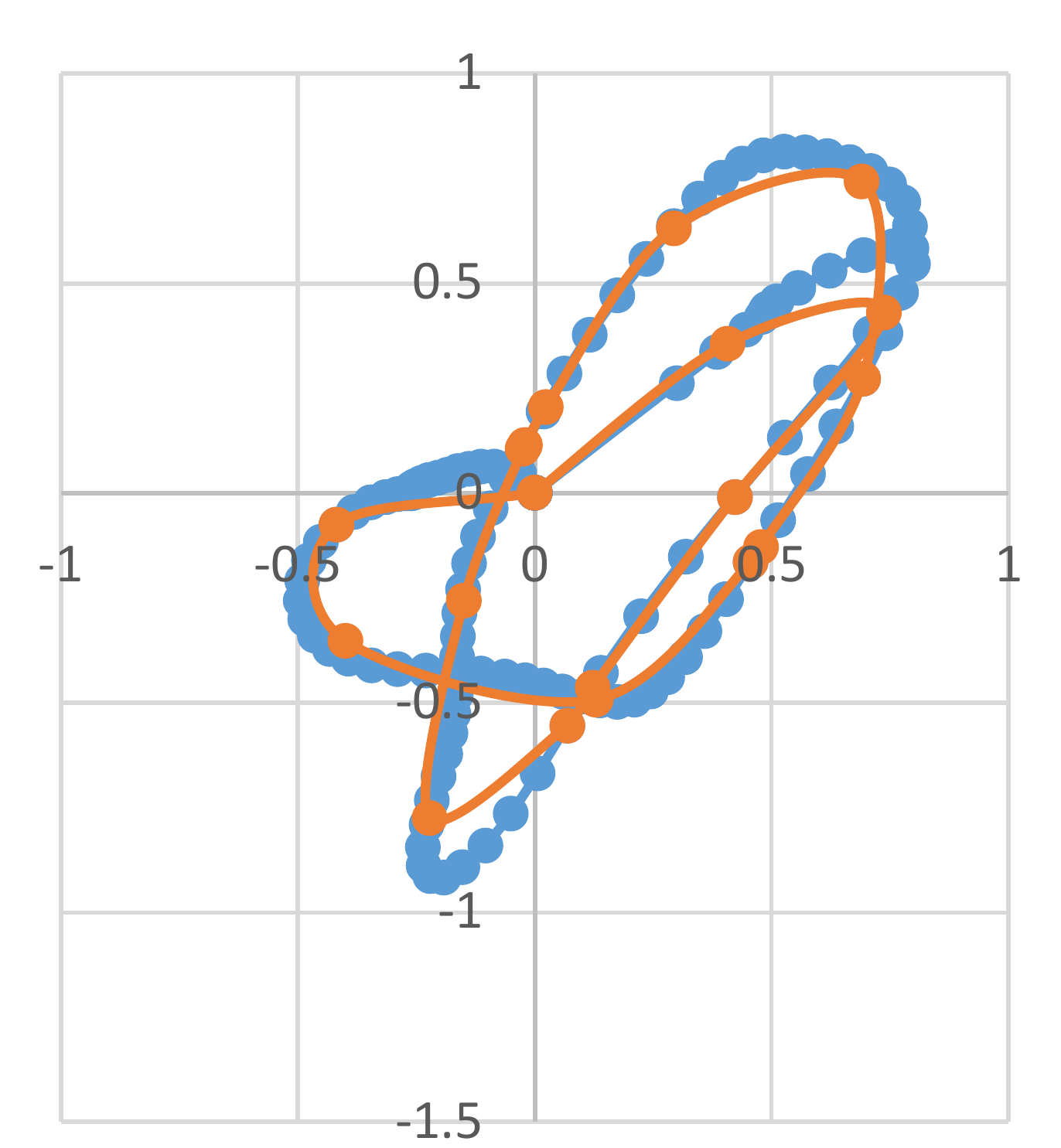}
\caption{The first sample (blue) and its simplification using the points of the plane (orange).}
\label{fig:char1}
\end{subfigure}
\begin{subfigure}{.45\textwidth}
\centering
\includegraphics[scale=0.2,trim={3cm 4cm 3cm 4cm},clip]{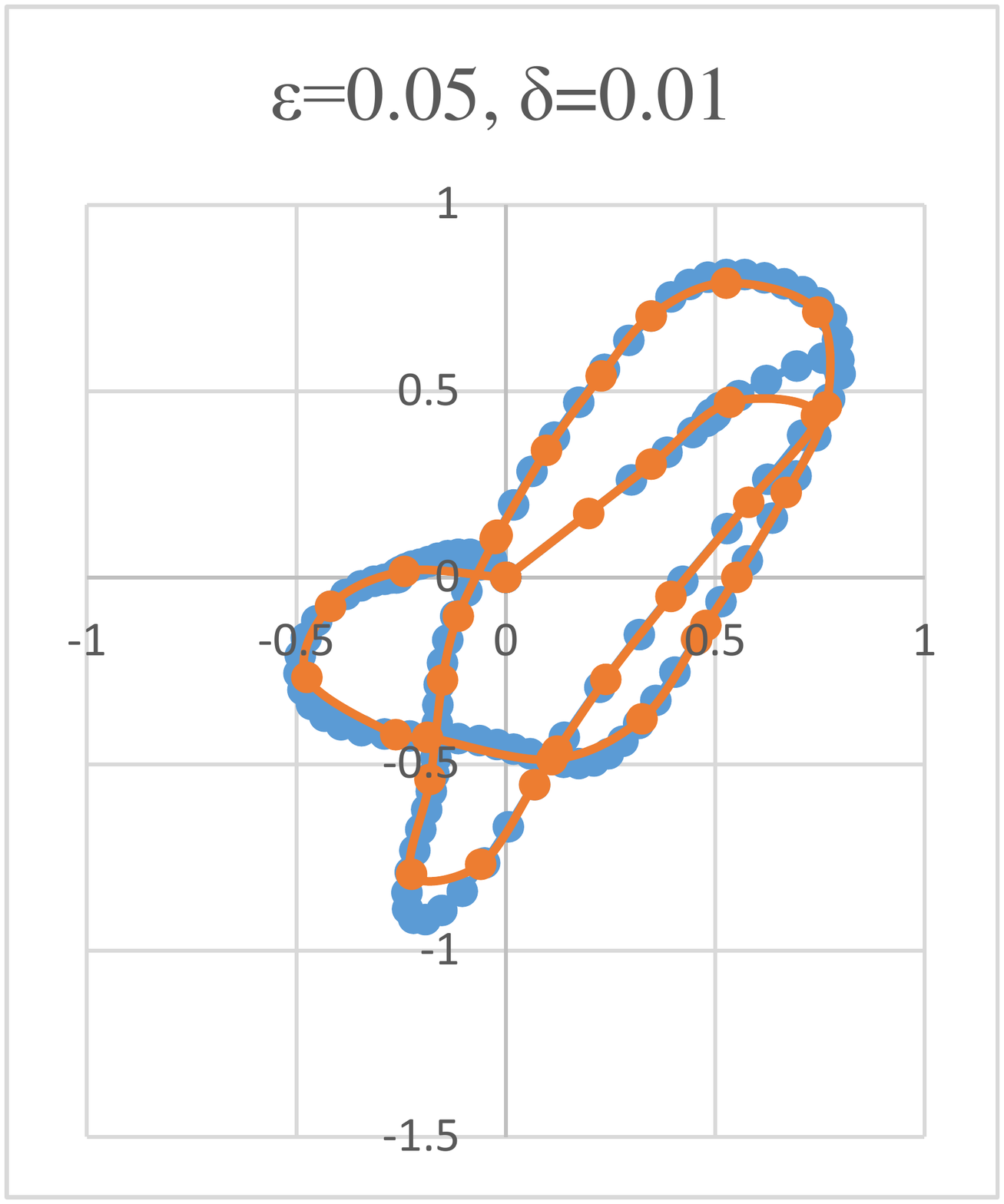}
\caption{The first sample (blue) and its simplification using the points of the plane (orange).}
\label{fig:parts}
\end{subfigure}\hfill
\begin{subfigure}{.45\textwidth}
\centering
\includegraphics[scale=0.2,trim={3cm 4cm 3cm 4cm},clip]{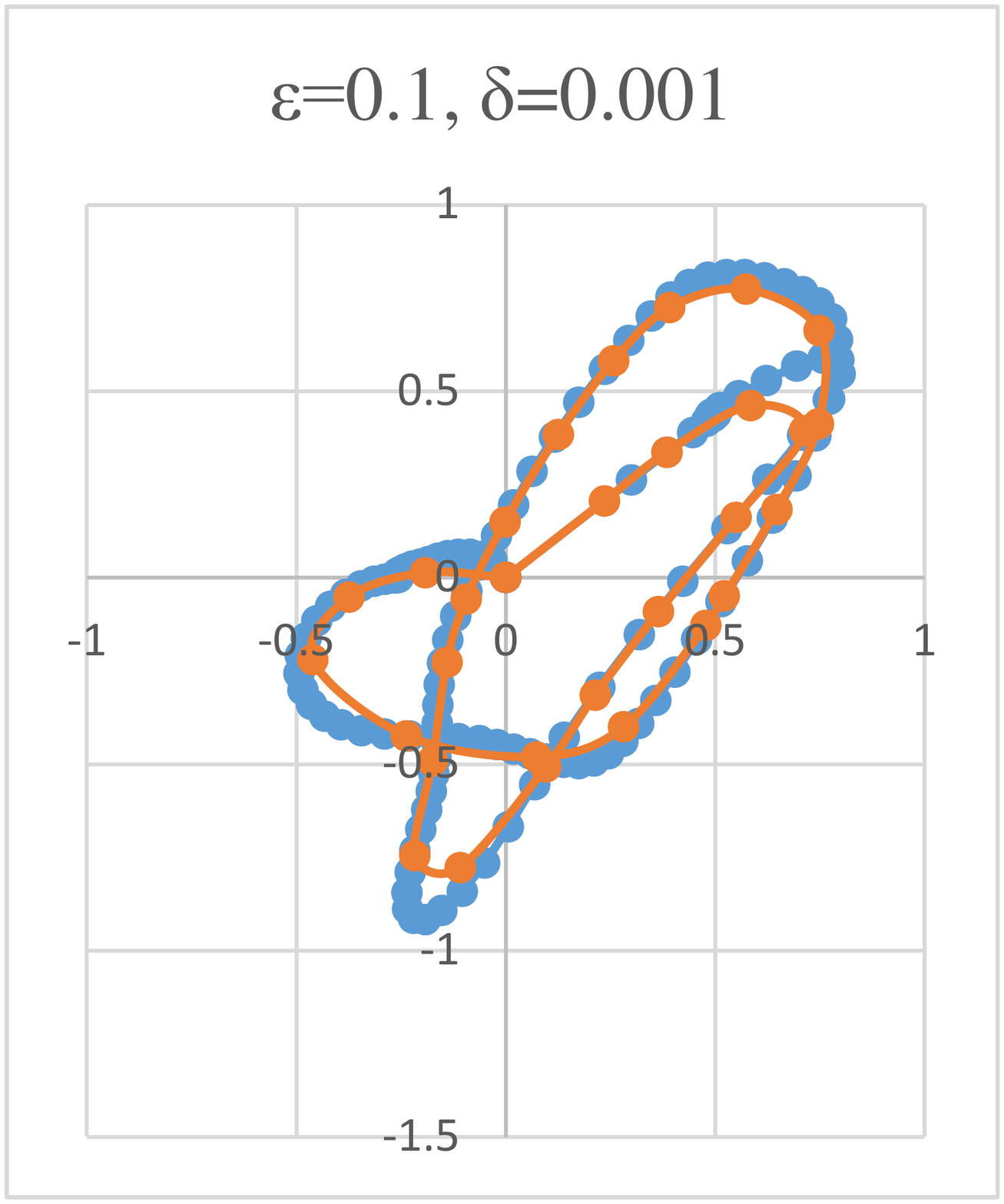}
\caption{The first sample (blue) and its simplification using the points of the plane (orange).}
\label{fig:merged}
\end{subfigure}
\end{figure}
\section{Discussions and Open Problems}
Using the free-space diagram to map $L$ two dimensional points into $L$ curve-length variables removes useful information such as the slope (line inclination). In this paper, we added this information by adding some vertices to the original definition of FSD and used it to give a simplification algorithm. Our algorithm generalizes to $L_p$ norm of the Fr\'echet distances, not to be confused with using $L_p$ norms as the metric space instead of the Euclidean plane.

We also show that the $p$-mean curve satisfies the composable property for core-sets, and results in a constant factor approximation summary.
\section{Theoretical Insights to Existing Heuristics}
\paragraph{Greedy simplification by moving a disk along the curve.}
The heuristic simplification algorithm that sweep the curve and simplifies the part of the curve that is inside the disk of radius $\epsilon$ is a commonly used algorithm in practice, which does not have any theoretical guarantees except for the error $\epsilon$. When discussing this simplification algorithms in the parameter space (FSD) of the curve with itself with $\epsilon$ as input, we see it is in fact the lowermost feasible path. Replacing this path with the shortest path gives an algorithm for min-$k$ simplification with complexity dependent on the size of the FSD for $\epsilon$.
For monotone curves, this is an output-sensitive exact algorithm that takes $O(n\log n+nk)$ time.

\paragraph{Local polyline simplification using Fr\'echet distance}
The results of \cite{optimal} showed the famous Imai-Iri algorithm~\cite{imaiiri} is not optimal for the simplification under Fr\'echet distance, when taken over the whole curves (global Fr\'echet distance).

\bibliographystyle{abbrv}
\bibliography{abbrv,refs}
\end{document}